\documentclass[12pt]{article}
\usepackage{amssymb}
\usepackage{array}
\usepackage{amsmath}
\usepackage{color}
\usepackage[colorlinks,linkcolor=blue,citecolor=green]{hyperref}
\newtheorem{theorem}{Theorem}

\newtheorem{example}[theorem]{Example}

\newenvironment{proof}[1][Proof]{\textbf{#1.} }{\ \rule{0.5em}{0.5em}}

\usepackage{graphicx}
\allowdisplaybreaks[3]
\usepackage{graphicx} 

\title{Quasilinear differential constraints for parabolic systems of Jordan-block type}
\author{\vspace{3mm}Alessandra Rizzo$^1$, Pierandrea Vergallo$^{2,3}$ \\ 
\small $^1$  Department of Mathematical, Computer,
\small  Physical and Earth Sciences\\
\small  University of Messina, \\
\small    \texttt{alessandra.rizzo1@unime.it} \\
\small $^2$  Department of Engineering,\\ \small University of Messina, \\
\small  
\texttt{pierandrea.vergallo@unime.it}\\
\small $^3$ {Istituto Nazionale di Fisica Nucleare, Sez.\ Milano} }
\date{}

\begin{document}
\maketitle

\vspace{10mm}

\begin{abstract}
    We prove that linear degeneracy is a necessary conditions for systems in Jordan-block form to admit a compatible quasilinear differential constraint. Such condition is also sufficient for $2\times 2$ systems and turns out to be equivalent to possess the Hamiltonian property.  Some explicit solutions of parabolic systems are herein given: two principal hierarchies arising from the associativity theory and the delta-functional reduction of the El's equation in the hard rod case are integrated.

    \vspace{5mm}

\end{abstract}

\newpage 
\section{Introduction}
Modern physics makes an extensive use of nonlinear Partial Differential Equations (PDEs) to model a large number of phenomena \cite{Logan,Scott}.  Here, we investigate nonlinear systems of $n$ first order PDEs which are homogeneous and quasilinear (earlier known as systems of \emph{hydrodynamic type}):
\begin{equation}
    \label{parsys}
    u^i_t=A^i_j(\textbf{u})u^j_x, \qquad i=1,2,\dots n,
\end{equation}
 where $t$ and $x$ are the independent variables, $u^i(t,x)$ ($i=1,2,\dots n$) are the field variables and the matrix $A^i_j$ depends on $\textbf{u}=(u^1,\dots u^n)$, but not on higher order derivatives. 
 
Systems \eqref{parsys} are commonly classified with respect to the eigenvalues of $A$. In particular, we say \eqref{parsys} to be hyperbolic if the eigenvalues are real with $n$ linearly independent eigenvectors; whereas it is said parabolic if it admits multiple real eigenvalues with one eigenvector for each; and finally, elliptic if its eigenvalues are complex.  A huge literature has been developed with regard to hyperbolic quasilinear systems, see for example \cite{Serre} for systems of conservation laws, \cite{DN83} for the study of the Hamiltonian property and \cite{Tsarev,Tsarev1} for their integrability. However, as far as the authors know, some aspects of the other cases still remain to be understood.  

In this paper, we focus on quasilinear systems of parabolic type, where we assume that the matrix $A$ has upper-triangular Toeplitz form,
\begin{equation}\label{J}
A=\lambda^1 \mathbb{I}+\sum_{i=1}^{n-1}\lambda^{i+1}P^i;
\end{equation}
here  $\mathbb{I}$ is the $n\times n$ identity matrix, $P$ is the $n\times n$ Jordan block with zero eigenvalue (note that $P^n=0$), and $\lambda^i$ are  functions of $\textbf{u}$. 
 Such systems are non-diagonalizable whenever $P^i\neq 0$. Our interest on this case arises from a number of applications recently shown in different fields: in the parabolic regularisation of the Riemann equation \cite{KO2}, in the reductions of hydrodynamic chains and linearly degenerate dispersionless PDEs in 3D \cite{P2003}, in Nijenhuis geometry \cite{BKM},  in the context of non-semisimple prepotential in WDVV theories \cite{LP} and in the novel developments on the El's kinetic equation under delta-functional reduction \cite{PTE, FP} and their Hamiltonianity \cite{VerFer1,VerFer2}. 

Let us conclude this introduction remarking that linearly degenerate systems play a key role in our framework, especially from the point of view of solvability of the initial value
problem \cite{R1}. For this reason, we recall that equations \eqref{parsys} in the strictly hyperbolic case are said linearly degenerate if each eigenvalue (also known as \emph{characteristic speed}) is constant along the direction of the corresponding eigenvector. In \cite{Fer}, a criterion of linear degeneracy was introduced that makes not explicit use of eigenvectors and eigenvalues, but is written in term of its characteristic polynomial
$$
P(\lambda)=det(\lambda \mathbb{I}-A)={\lambda  }^n  +  f_1(\textbf{u}){\lambda
}^{n-1} +f_2({\textbf{u}}){\lambda}^{n-2}+ \ldots + f_n({\textbf{u}}).
$$
Then, the condition of linear degeneracy can be equivalently expressed in the following form (see \cite{Fer})
\begin{equation}\label{lindeg}
\nabla f_1~A^{n-1}+\nabla  f_2~A^{n-2}+\ldots  +\nabla  f_n=0,
\end{equation}
where   $\nabla f=(\frac{\partial
f}{\partial u^1},\ldots , \frac{\partial f}{\partial u^n})$ is the gradient, and $A^k$ denotes  $k$-th power of the matrix $A$.  The present criterion is taken here (but see also \cite{VerVit2,VerFer1}) as definition for quasilinear systems not only in the strictly hyperbolic case.  

It turns out that for systems in Jordan-block form the following result holds:
 \begin{theorem}\label{thm0} { Suppose that the matrix $A$ has block-diagonal form with several  blocks  $J_{\alpha}$ of type (\ref{J}) with distinct eigenvalues $\lambda^1_{\alpha}$. Then the condition of linear degeneracy is equivalent to 
\begin{equation}\label{l1}
\frac{\partial \lambda^1_{\alpha}}{\partial u^1_{\alpha}}=0 \qquad \forall \alpha, 
\end{equation}
no summation}.\end{theorem}

This result was firstly shown in \cite{VerFer1} and used therein to prove that linearly degeneracy is a necessary property for system in Jordan-block form to be Hamiltonian. We will make a remarkable use of this Theorem also in the present paper. 
 
\vspace{5mm}

The focus of our study is double: we want to find solutions to some parabolic quasilinear systems in Toeplitz form by using the method of differential constrains and to relate the applicability of such a reduction procedure to intrisic geometric properties of the system, as its linear degeneracy and Hamiltonianity. To these aims, in Section \ref{2} we describe the general framework behind the method of differential constraints and we analyse systems in $n=2$ and $n>2$ number of components. In the present case, it turns out that linear degeneracy is a necessary property to apply such a procedure (as well as to be Hamiltonian). In Section \ref{4}, we use the obtained results to compute exact solutions for physical examples (arising in the contest of WDVV equations and the kinetic equations for soliton gas theories). Finally, in Section \ref{5} we show some preliminary results on the connection between the Hamiltonian property of the systems and the underlined method.

\section{The method of differential constraints for systems in Toeplitz form}\label{2}

Searching for solutions of nonlinear equations is not an easy task in general. Different methods and algorithms appeared in the literature at this aim, making use of algebro-geometric approaches to the system. We refer to \cite{2a} for a complete reference on the subject. Here, we focus on the method of differential constraint, that was firstly introduced by N.N. Yanenko in 1964 \cite{1a} and has shown an increasing number of applications in the field of nonlinear waves \cite{2a,3a,4a,5a,ManRizVer1, ManRiz4, ManRiz1}. In the following, we briefly recall the main ideas and steps of the present method. 

Let $u_t^i=f^i(\textbf{u},\textbf{u}_x,\dots, \textbf{u}_{kx})$ be a given system of $N$ evolutionary equations (note that this can be easily generalized to arbitrary PDEs). The idea of the method is simply to append to the system a set of $M$, with $M<N$, differential equations
$\Phi^j(x,t,\textbf{u},\textbf{u}_x,\dots, \textbf{u}_{kx})=0$, which play the role of constraints, selecting a class of exact solutions of the initial system.  In fact, we look for the exact particular solutions of the original systems which satisfy also the appended differential constraints. After choosing the form of the constraints, one needs to solve the compatibility condition among the initial system and the new differential relations, requiring the involutivity of the obtained overdetermined system. In \cite{Zhi}, it was proved that for strictly hyperbolic quasilinear system,
 the more general first-order  differential constraints that can be appended to the system are quasilinear and take the form
\begin{equation}
 l_{j}^{i} u_{x}^{j}= p^{i},
\end{equation}
where $l^{i}$ are the left eigenvectors of the initial system and $p^i$ are source functions to be determined according to the compatibility analysis.\\
Once that the consistency conditions are satisfied, one can look for solutions of the starting system with the auxiliary differential conditions that commonly simplify the computations required. 

We now specify the application of this procedure to parabolic system in Jordan-block form, dividing our investigation in two subcases: $n=2$ and $n>2$. 

\subsection{Systems in $n=2$ components} 
Let us start from the easiest case of $2\times 2$ systems of Toeplitz form in the field variables $u^1$,$u^2$, i.e. 
\begin{equation}\label{qsj2}
\begin{cases}u^1_t=\lambda({u^1},{u^2})u^1_x+\mu({u^1},{u^2})u^2_x\\u^2_t=   \hphantom{\lambda({u^1},{u^2})u^1_x+} \lambda({u^1},{u^2})u^2_x
\end{cases}
\end{equation} where $\lambda=\lambda(u^{1},u^{2})$ and $\mu=\mu(u^{1},u^{2})$ are arbitrary functions.  In this case, we add to the system \eqref{qsj2} the nonlinear differential constraint
\begin{equation}\label{74}
    \Phi(x,t,u^1,u^2,p,q)=0,
\end{equation}
where $p=u^1_x$, $q= u^2_x$ and we also define $r=u^1_t,s=u^2_t$. In \eqref{74} we reasonably required the constraint to be of the same order of the system, i.e. $\Phi$ does not depend on derivatives $\textbf{u}_{kx}$ with $k>1$. In this case, we prove that only quasilinear constraints are admitted:

\bigskip

\begin{theorem}\label{01}
    For $2\times 2$ systems in Jordan-block form \eqref{qsj2}, the differential constraint \eqref{74} is quasilinear:
    \begin{equation*}
        0=\Phi(x,t,u^1,u^2,p,q)=u^2_x-\varphi(x,t,u^1,u^2)
    \end{equation*}
    \label{thmn2}
\end{theorem}
\begin{proof}
    In order to integrate the system we look for a constraint 
    \begin{equation}
        \Phi(x,t,u^{1}, u^{2},p,q)=0.
        \label{constraint}
    \end{equation}
We want to study the compatibility between the constraint and the system, so we derive ($\ref{constraint}$) with respect to $x$ and $t$ and, taking into account the equations of the initial system, we obtain the following algebraic conditions for the second order derivatives:
\begin{subequations}
\label{compatibilityconditions1}
\begin{align*}
\label{compatibilityconditions}
&\Phi_x+ \Phi_{u^1} p+ \Phi_{u^2} q+ \Phi_p u^1_{xx}+ \Phi_q u^2_{xx}=0\\
&\Phi_t+ \Phi_{u^1} r + \Phi_{u^2} s+ \Phi_p u^1_{tx}+ \Phi_q u^2_{tx}=0\\
&u^1_{xt}- \lambda u^1_{xx}- \mu u^2_{xx}= \lambda_{u^{1}} p^{2}+ (\lambda_{u^{2}} +\mu_{u^{1}}) p q +\mu{u^{2}} q^{2}\\
&u^1_{tt}- \lambda u^1_{xt}- \mu u^2_{xt}= \lambda \lambda_{u^{1}} p^{2}+ (\mu  \lambda_{u^{1}}+ \lambda \lambda_{u^{2}} +\lambda \mu_{u^{1}}) p q+ (\mu \mu_{u^{1}}+\lambda \mu_{u^{2}}) q^2 \\
&u^2_{xt}- \lambda u^2_{xx}=\lambda_{u^{1}} p q+ \lambda_{u^{2}}  q^2\\
&u^2_{tt}- \lambda u^2_{xt}=\lambda \lambda_{u^{1}} p q+(\mu \lambda_{u^{1}}+ \lambda \lambda_{u^{2}} ) q^{2} 
\end{align*}
\end{subequations}
Since each solution of $ (\ref{constraint}) $ depends on one arbitrary function we require the same level of arbitrariness for the solutions of the system above. Hence, we need to require that the matrix $D$ associated to the second order derivatives has rank less than 6. This leads to the condition
\begin{equation*}
   0=\det D = \mu \Phi_{p}^2.
\end{equation*}
Then we have that our constraint cannot depend on $u_{x}^{1}$, so, without loss of generality, we can rewrite \eqref{74} as 
\begin{equation*}
    u^2_x= \varphi(x,t,u^1,u^2). \label{formavincolo}
\end{equation*}
\end{proof}

\bigskip

The previous Theorem let us deduce that, as in the hyperbolic case, quasilinear differential constraints play a key role for hydrodynamic type systems.  As a result, we use Theorem \ref{01} to prove the following

\bigskip

\begin{theorem}\label{thm3}
    Linear degeneracy is a necessary and sufficient condition for quasilinear systems in Jordan-block form \eqref{qsj2} to admit a differential constraint.
\end{theorem}
\begin{proof}
To prove this result, we complete the study of compatibility between the constraint and the system simply by requiring that the Schwartz condition
\begin{equation*}
    u_{xt}^{2}=u_{tx}^{2}
\end{equation*}
is satisfied. Then, we get a polynomial in $u_{x}^{1}$ that is satisfied for each choice of $u^1$ if and only if
\begin{subequations}
\begin{align}
    &\lambda_{u^{1}}=0 \label{10a} \\
     & \lambda_{u^2} (\varphi^1)^2+ \lambda \varphi_{x}^{1}= \varphi_{t}^{1}+ \mu \varphi^1\varphi_{u^1}^{1}. \label{secondcondition}
\end{align}
\end{subequations}
Hence if the system is linearly degenerate, it admits a differential constraint whose form is given by the integration of \eqref{secondcondition}. Viceversa, every quasilinear differential constraint satisfies condition (\ref{10a}, \ref{secondcondition}) and $\lambda$ must depend only on $u^2$, resulting to be linearly degenerate. \end{proof}  

\bigskip

The present result shows a concrete connection between the linear degeneracy property and the applicability of the method of differential constraints. In addition, we will see in Section \ref{5} that this result has an analogue when investigating the Hamiltonian structures in $n=2$ Jordan-block systems.

We finally remark that in the particular case when $\varphi_{x}^{1}=\varphi_{t}^{1}=0$, the second condition simply becomes
\begin{equation*}
    \lambda_{u^2}(u^2)=  \mu(u^1,u^2) \dfrac{\varphi^1_{u^1}}{\varphi^1},  
\end{equation*} so that 
\begin{equation}\label{formt}
    \varphi^1(u^1,u^2)=f^1(u^2) \: e^{\displaystyle \int{\frac{\lambda_{u^2}}{\mu}\, du^1}},
\end{equation}where $f^1(u^2)$ is an arbitrary function. We will show in the next Section that the requirement that $\varphi^{1}$ does not depend explicitly on $t$ and $x$ turns out to be useful in the applications, that is searching for solutions of the systems.

\subsection{Systems in $n>2$ components}\label{subse2}
In the previous sub-section we considered the case of $2 \times 2$ systems and we proved that differential constraints are strictly related to linearly degenerate property.
Here, we investigate the differential constraints for systems of $l>1$ blocks of the form
\begin{equation}
    \label{qldc}
u^{k_\alpha}_{\alpha,x}=\varphi^\alpha(x,t,\textbf{u}), \qquad \alpha=1,2,\dots l,
\end{equation}
where we indicate with $k_\alpha$ the length of the $\alpha$-th block and $u^{k_\alpha}_\alpha$ the $k_\alpha$-th field variable in the $\alpha$-th block. As an example, in the system
\begin{equation*}
    \begin{pmatrix}
        u^1_1\\u^1_2\\u^2_2\\u^3_2
    \end{pmatrix}_t=
    \begin{pmatrix}
        \lambda^1&0&0&0\\
        0&\lambda^2&\mu^2&\eta^2\\
        0&0&\lambda^2&\mu^2\\0&0&0&\lambda^2
    \end{pmatrix}\begin{pmatrix}
         u^1_1\\u^1_2\\u^2_2\\u^3_2
    \end{pmatrix}_x,
\end{equation*}
there are two Jordan-blocks $i=1,2$ of length respectively $k_1=1,k_2=3$, so that we indicate with $u^{1}_1=u^1$ and $u^{3}_2=u^4$ and we consider the quasilinear constraints with respect to the last variables of each block
$$u^1_{1,x}=u^1_x=\varphi^1(x,t,\textbf{u}),\qquad u^3_{2,x}=u^4_x=\varphi^2(x,t,\textbf{u}).$$

\medskip 

Finally, we consider the case of parabolic systems composed by different Jordan blocks of arbitrary dimensions.  In what follows, we will consider a differential constraint $\Phi^\alpha$ for each Jordan block $\alpha$ of $A$. An analogue (but weaker) result of Theorem \ref{thm3} is obtained.

\medskip

\begin{theorem}\label{thm4}
    Linear degeneracy is a necessary condition for quasilinear systems of Jordan-block type \eqref{J} to admit quasilinear differential constraints \eqref{qldc}.
\end{theorem}
\begin{proof}
    Let us consider for sake of simplicity a system with 2 Jordan blocks of length $k,m$ and eigenvalues $\lambda_1=\lambda_{1}^{1},\lambda_2=\lambda_{2}^{1}$ respectively. The proof can be easily generalized for an arbitrary number blocks. In this case, the system reads as
    \begin{equation*}
    \begin{cases}\displaystyle 
         u_{1,t}^{i}= \sum_{j=i}^{k} \lambda_1^{j-i+1}(\textbf{u})\, u_{1,x}^{j}  \mathrm{\: } \qquad i=1, \dots ,k\\
        \displaystyle u_{2,t}^{s}= \sum_{j=s}^{m} \lambda_2^{j-i+1}(\textbf{u})\, u_{2,x}^{j}\qquad s=1, \dots ,m.
     \end{cases},
      \end{equation*}along with the constraints
      \begin{equation}\label{questiqui}
          u^k_{1,x}=\varphi^{1}(t,x,u^1,\dots u^{k+m}), \qquad u^m_{2,x}=\varphi^2(t,x,u^1,\dots u^{k+m}).
      \end{equation}Note that $k+m=n>2$. We finally require the compatibility conditions \begin{align*}
     u_{1,tx}^{k}= u_{1,xt}^{k}, \qquad 
      u_{2,tx}^{m}=u_{2,tx}^{m} .
\end{align*} to be satisfied. Computing the total derivatives and substituting $u^i_t=A^i_ju^j_x$ and the constraints \eqref{questiqui}, we get a polynomial in the first order derivatives of ${u}^{1}, {u}^{2}$. We get the compatibility conditions

\smallskip 

     \begin{subequations}\label{22}
\begin{align}
  \displaystyle  &\frac{\partial \lambda^1_1}{\partial u_1^1}  =\frac{\partial \lambda^1}{\partial u^1}=0,\label{eq111}\qquad  \frac{\partial \lambda^1_2}{\partial u^1_2}=\frac{\partial \lambda^2}{\partial u^{k+1}}=0,\end{align}
  \medskip
  
  with the following additional relations
  \begin{align}
   \displaystyle &\sum_{l=1}^{j} \dfrac{\partial \varphi^{1}}{\partial u_{1}^{j-l+1}} \lambda_{1}^{l}=\dfrac{\partial \lambda_{1}^{1} \varphi^{1}}{ \partial u_{1}^{j}} \qquad \mathrm{for \: each \:} j=2, \dots, k-1, \vspace{0.2 cm} \\
    \displaystyle &\sum_{l=1}^{s} \dfrac{\partial \varphi^{1}}{\partial u_{2}^{s-l+1}} \lambda_{2}^{l}=\dfrac{\partial \lambda_{1}^{1} \varphi^{1}}{ \partial u_{2}^{s}} \qquad \mathrm{for \: each \:} s=1, \dots, m-1,\\
\displaystyle  &\varphi_{t}^{1}+   \sum_{l=1}^{k}  \dfrac{\partial \varphi}{\partial u_{1}^{k-l+1}} \lambda_{1}^{l} \varphi^{1}+ \sum_{l=1}^{m} \dfrac{\partial \varphi^{1}}{\partial u_{2}^{m-l+1}}  \lambda_{2}^{l} \varphi^{2}=  \varphi^{1} \left(\dfrac{\partial \lambda_{1}^{1} \varphi^{1}}{ \partial u_{1}^{k}} \right)+ \varphi^{2} \left(\dfrac{\partial \lambda_{1}^{1} \varphi^{1}}{ \partial u_{2}^{m}} \right)+ \lambda_{1}^{1} \varphi_{x}^{1},\\
  \displaystyle &\sum_{l=1}^{j} \dfrac{\partial \varphi^{2}}{\partial u_{1}^{j-l+1}} \lambda_{1}^{l}=\dfrac{\partial \lambda_{2}^{1} \varphi^{2}}{ \partial u_{1}^{j}} \qquad \mathrm{for \: each \:} j=2, \dots, k-1, \vspace{0.2 cm} \\
 \displaystyle  &\sum_{l=1}^{s} \dfrac{\partial \varphi^{2}}{\partial u_{2}^{s-l+1}} \mu^{l}=\dfrac{\partial \mu^{1} \varphi^{2}}{ \partial  u_{2}^{s}} \qquad \mathrm{for \: each \:} s=1, \dots, m-1,\\
  \displaystyle  &\varphi_{t}^{2}+   \sum_{l=1}^{k}  \dfrac{\partial \varphi^{2}}{\partial u_{1}^{k-l+1}} \lambda_{1}^{l} \varphi^{1}+ \sum_{l=1}^{m} \dfrac{\partial \varphi^{2}}{\partial  u_{2}^{m-l+1}}  \lambda_{2}^{l} \varphi^{2}= \varphi^{1} \left(\dfrac{\partial \lambda_{2}^{1} \varphi^{2}}{ \partial  u_{1}^{k}} \right)+ \varphi^{2} \left(\dfrac{\partial \lambda_{2}^{1} \varphi^{2}}{ \partial u_{2}^{m}} \right)+ \lambda_{2}^{1} \varphi_{x}^{2}.
\end{align}
\end{subequations}

\medskip

      In particular,  \eqref{eq111} implies the system to be linearly degenerate according to Theorem \ref{thm0}.

\end{proof}

\bigskip

At this point, we conclude this section remarking again that an analogous result in \cite{VerFer1} shows that linear degeneracy is also a necessary condition for quasilinear systems in Jordan-block form to be Hamiltonian with Dubrovin-Novikov operators. We remaind Section \ref{5} for further details.

\section{Applications}\label{4}
In this Section, we apply the reduction procedure described to obtain solutions to some examples. In particular, we will solve conditions \eqref{22} arising in the proof of the previous Theorem to integrate the given systems. 

\subsection{Systems in $2$ components} In \cite{XF}, L. Xue and E.V. Ferapontov observed that $2$-component systems in Toeplitz form are invariant in form under the change of variables 
\begin{equation*}
    u^1=F(u^1,u^2),\quad u^2=G(u^2),
\end{equation*}with $F$ and $G$ arbitrary functions of their arguments. Then, owing to Theorems \ref{thm0} and \ref{thm3}, a compatible differential constraint requires $\lambda=\lambda(u^2)$ so that the system \eqref{qsj2} is linearly degenerate and, in this case, we can map it to the following
\begin{equation*}
\begin{cases}
u^1_t=u^2u^1_x+\mu(u^1,u^2) \,u^2_x\\
u^2_t=u^2u^2_x
\end{cases}
\end{equation*}
Moreover, by applying the transformation $u^1=f(r^1,r^2)$ and $u^2=r^2$ the system becomes
\begin{equation*}
\begin{cases}f_{r^1}r^1_t+f_{r^2}r^2_t=r^2(f_{r^{1}}r^1_x+f_{r^2}r^2_x)+\mu \, r^2_x\\r^2_t=r^2r^2_x
\end{cases}
\end{equation*}Finally, dividing by $f_{r^1}$ and choosing $f_{r^1}=\mu$ the system is mapped into the following system where we choose for sake of simplicity the notations in $u^1$ and $u^2$ (see also \cite{XF}):
\begin{equation}\label{1}
\begin{cases}u^1_t=u^2u^1_x+u^2_x\\u^2_t=u^2u^2_x\end{cases}
\end{equation}

Therefore, without loss of generality, we can consider system \eqref{1}
along with the quasilinear constraint (according to Theorem \ref{thmn2} and formula \eqref{formt}) 
\begin{equation}
    u^2_x=\varphi^1(u^1,u^2)= f^1(u^2) \, e^{u^1}.
    \label{constraintn2}
\end{equation} We now firstly consider the second equation of the system. Through the method of characteristics, we obtain
\begin{equation*}
    u^2=u_{0}^{2}(\sigma), \qquad \sigma= x+u_{0}^{2}(\sigma) t.
\end{equation*}
In order to study the first equation of the system, we use the change of variables  $(x,t) \to (\sigma, \tau)$, where
$x= \sigma- u_{0}^2(\sigma) \tau$ and
        $t= \tau$.
Hence, we obtain
\begin{equation*}
u^1_\tau= f^1(u_{0}^{2}) \, e^{u^1}  .
\end{equation*}
By integration, we get
\begin{equation}\label{sol1}
  u^1(\sigma,\tau)=\log{\left(\dfrac{1}{e^{-u_{0}^{1}(\sigma)}- f^1(u_{0}^{2}) \, \tau}\right)},
\end{equation}
where $u_{0}^{1}(\sigma)$ is an arbitrary function in its argument. At the end, substituting in the constraint \eqref{constraintn2} we obtain that the functions $f^{1}(u_{0}^{2}), \, u_{0}^{1}, \, u_{0}^{2}$ satisfy
\begin{equation}\label{sol2}
  \dfrac{u_{0}^{2'}}{f^{1}(u_{0}^{2})}= e^{u_{0}^{1}}.
\end{equation}
Hence, we are able to provide a solution up to two arbitrary functions. This means that equations \eqref{sol1} along with \eqref{sol2} give the general integral of the system. In this way, we found a complete set of solutions of every Jordan-block system in Toeplitz form gettable via the method of differential constraint. 
\subsection{Systems in higher number of components}
We now apply the results of Subsection \ref{subse2} to systems of Jordan-block type in Toeplitz form in $n=3$ and $n=4$ components. The examples chosen come from different areas of physics, showing  the range of applicability of the method. 

Here, we use conditions \eqref{22} computed in the proof of Theorem \ref{thm3} and specify $\lambda^i$ for each case of study. 

\subsubsection{Principal hierarchies of Frobenius manifolds}
Let us consider an example of an integrable hierarchy of Jordan block type coming from the theory of associativity equations, also known as WDVV (Witten-Dijkgraaf-Verlinde-Verlinde) equations \cite{wdvv}. These equations firstly appeared in the contest of 2D topological field theory but they were widely studied, see e.g. \cite{DubFiel,wdvv1,wdvv2}.

A key role in the associativity equations is played by a function $F$ called \emph{pre-potential}. Here we consider a non-semisimple WDVV prepotential as follows 

$$
F(u^1,u^2,u^3)=\frac{1}{2}(u^1)^2u^3+\frac{1}{2}u^1(u^2)^2+\frac{1}{8}\frac{(u^2)^4}{u^3}.
$$
It appeared in \cite{PVV}, Section 6, and also  \cite{LP}, eqn (3.18). A standard procedure associates to $F$ two commuting systems (primary flows in the language of WDVV equations),

\begin{align*}
u^1_t=(F_{u^2u^3})_x, \quad u^2_t=(F_{u^2u^2})_x, \quad u^3_t=(F_{u^1u^2})_x, 
\\
u^1_s=(F_{u^3u^3})_x, \quad u^2_s=(F_{u^2u^3})_x, \quad u^3_s=(F_{u^1u^3})_x. 
\end{align*}   

which explicitly read as

\begin{equation*}
u^1_t=-\frac{3}{2}\frac{(u^2)^2}{(u^3)^2}u^2_x+\frac{(u^2)^3}{(u^3)^3}u^3_x, \quad u^2_t=u^1_x+3\frac{u^2}{u^3}u^2_x-\frac{3}{2}\frac{(u^2)^2}{(u^3)^2}u^3_x, \quad u^3_t=u^2_x, 
\end{equation*}
\begin{equation*}
u^1_s=\frac{(u^2)^3}{(u^3)^3}u^2_x-\frac{3}{4}\frac{(u^2)^4}{(u^3)^4}u^3_x, \quad u^2_s=-\frac{3}{2}\frac{(u^2)^2}{(u^3)^2}u^2_x+\frac{(u^2)^3}{(u^3)^3}u^3_x, \quad u^3_s=u^1_x. 
\end{equation*}

\medskip

Following \cite{VerFer1}, we apply the transformation of field variables $$u^1=-\frac{1}{u^3}, \qquad u^2=\frac{u^2}{u^3}, \qquad u^3=u^1+\frac{1}{2}\frac{(u^2)^2}{u^3},$$ such that these commuting systems are mapped into an upper-triangular Toeplitz form,
\begin{equation}\label{E1}
\left(\begin{array}{c}u^1\\ u^2\\u^3\end{array}\right)_t=
\left(\begin{array}{ccc}
u^2 &-u^1 & 0\\
0&u^2&-u^1\\
0&0&u^2
\end{array}\right)
\left(\begin{array}{c}u^1\\ u^2\\ u^3\end{array}\right)_x
\end{equation}
and 
\begin{equation}\label{sys2}
\left(\begin{array}{c}u^1\\ u^2\\u^3\end{array}\right)_s=
\left(\begin{array}{ccc}
-\frac{1}{2}(u^2)^2 &u^1u^2 & (u^1)^2\\
0&-\frac{1}{2}(u^2)^2 &u^1u^2\\
0&0&-\frac{1}{2}(u^2)^2 
\end{array}\right)
\left(\begin{array}{c}u^1\\ u^2\\ u^3\end{array}\right)_x.
\end{equation}

\vspace{3mm}

Both the systems are integrable, in the sense that they admit an infinite number of compatible Hamiltonian structures of Dubrovin-Novikov type (see \cite{VerFer1}). We now apply the reduction procedure of differential constraints to find solutions.

\bigskip 

\textbf{ System \eqref{E1}.}
          Let us consider equations in \eqref{E1} and the quasilinear differential constraint
           $   u^3_x= \varphi^1(u^1,u^2,u^3)
           $, and by the compatibility conditions we obtain
           \begin{equation*}
               \varphi^1(u^1,u^2,u^3)= \dfrac{f^1(u^3)}{u^1},
           \end{equation*}
          where $f^1(u^3)$ is an arbitrary function. Substituting into the system, it yields to
          $$u^1_t-u^2u^1_x=-u^1u^2_x, \qquad u^2_t-u^2u^2_x=-f^1(u^3),\qquad u^3_t-u^2u^3_x=0,$$ 
    hence, we have
    \begin{subequations}
    \begin{align}
        &u^3(t,x)= u_{0}^{3}(\sigma)
         \\
           &   u^2(t,x)=-f^1(u_{0}^{3}) t+ u_{0}^{2}(\sigma),
         \end{align}
         \end{subequations}
         where $u_{0}^{2}(\sigma),u_{0}^{3}(\sigma)$ are arbitrary functions depending on the variable $\sigma$ 
 implicitly defined by the characteristic equation 
         \begin{equation*}
             x=f^1(u_{0}^{3}(\sigma)) \dfrac{t^2}{2}- u_{0}^{2}(\sigma) t+ \sigma.
         \end{equation*}
         Let us now apply the change of variables $(x,t) \to (\sigma, \tau)$, where $\tau=t$.
         The equation for $u^{1}$ becomes 
      \begin{equation*}
          \dfrac{u^1_\tau}{u^1}=\dfrac{f^{1'} u_{0}^{3'}  \tau - u_{0}^{2'} }{\dfrac{ f^{1'}}{2} u_{0}^{3'}  \tau^2- u_{0}^{2'}  \tau+1}.
      \end{equation*}
    Integrating with respect 
 to $\tau$
      \begin{equation}
          u^1(t,x)=\log \left(\dfrac{f^{1'}}{2} u_{0}^{3'} \tau^2- u_{0}^{2'} \tau+1 \right)+ u_{0}^{1}(\sigma).
      \end{equation}
      Substituting into the constraint, we have the additional condition
      \begin{equation}
          f^1(u_{0}^{3})= \dfrac{u_{0}^{3'} }{u_{0}^{1} },
      \end{equation}
     that is, given an initial datum, we can compute $f^1$.  We remark that there exist three arbitrary functions $u_{0}^{1}, \, u_{0}^{2}, \, u_{0}^{3} $ entering into the solution of the system. Then, we found a general integral of the equations.

\bigskip
      
   \textbf{System \eqref{sys2}.}
   We consider system \eqref{sys2}
          along with the quasilinear constraint $u^3_x= \alpha(u^1,u^2,u^3)$. From the compatibility conditions, we get
\begin{equation*}
    \varphi^1(u^1,u^2,u^3)= \dfrac{u^2}{u^1}f^1(u^3),
\end{equation*}
where $f^1(u^3)$ is an arbitrary function.
So that, we can rewrite the system as
\begin{align*}\begin{split}
&u^1_t+\frac{1}{2}(u^2)^2u^1_x=u^1u^2u^2_x+u^1u^2\,f^1(u^3),
\\&u^2_t+\frac{1}{2} (u^2)^2 u^2_x=(u^2)^2\,f^1(u^3), \qquad u^3_t+\frac{1}{2} (u^2)^2 u^3_x=0
\end{split}\end{align*}

          By integration along the characteristics, we obtain
          \begin{subequations}
          \begin{align}
              &u^3=u_{0}^{3}(\sigma)\\
               &   u^2= \dfrac{u_{0}^{2}(\sigma)}{1- f^1  \, u_{0}^{2}(\sigma) t},
             \end{align}
             \end{subequations}
            where $u_{0}^{2}(\sigma),u_{0}^{3}(\sigma)$ are arbitrary functions. The variable $\sigma$ is defined by the characteristic equation           \begin{equation*}
                x= \dfrac{(u_{0}^{2})^{2} t}{2(1-f^{1} u_{0}^{2} t)}+ \sigma.
             \end{equation*}
             Finally, we introduce a change of variables $(x,t) \to (\sigma,\tau)$ (with $t=\tau)$ and in this way, taking also the constraint into account, we get to the following differential equation for $u^1$
             \begin{equation*}
                 u^1_\tau -\dfrac{f^{1} u_{0}^{2} }{1-f^{1} u_{0}^{2} \tau} u^{1}= \dfrac{[u_{0}^{2'} +(u_{0}^{2} ) ^2 \, \tau \, f^{1'} u_{0}^{3'}](u_{0}^{2} )^{2} f^{1}}{u_{0}^{3'}(1-f^{1}u_{0}^{2}  \tau)^4}
             \end{equation*}
             The solution for $u^1$ is given by
             \begin{equation}
                 u^1(t,x)= \dfrac{1}{1- f^{1} \tau u_{0}^{2}} \left\{
 u_{0}^{1}(\sigma) + \dfrac{(u_{0}^{2})^2 \, f^1 \, \tau [2 u_{0}^{2'}+ \tau u_{0}^{2} (f^{1'} u_{0}^{2} u_{0}^{3'}-f^{1}u_{0}^{2'})]}{2 u_{0}^{3'} (-1+f^{1}  u_{0}^{2} \tau)^2}\right\}.
             \end{equation}
            The functions $u_{0}^{1}, u_{0}^{2}$, $u_{0}^{3}$, and $f^{1}(u_{0}^{3})$ satisfy the constraint
            \begin{equation}
                f^{1}= \dfrac{u_{0}^{1}}{u_{0}^{2}}  u_{0}^{3'} ,
            \end{equation}
so that, also in this second case, we are able to produce a general integral of the system depending on the three arbitrary functions $u_0^1, u_0^2, $ and $u_0^3.$
 \medskip

\noindent \textbf{Remark.} In both the systems investigated, the solutions coming from our study depend on exactly $3$ arbitrary functions. This result could be useful in the applications: indeed different specific solutions can be found by fixing initial values for the field variables $u^1,u^2$ and $u^3$. 
     
\bigskip 

\subsubsection{Kinetic equation of soliton gas: the hard rod case}
The soliton gas theory is an emerging topic in Mathematical Physics due to its wide range of applications going from nonlinear optics and water waves to statistical mechanics, and has its foundations in the field of integrable systems \cite{ElRev}. A key role in this theory is played by the El's integro-differential kinetic equation describing the evolution of a dense soliton gas
\cite{El,EK, EKPZ}. In this section, we investigate the hard rod gas case, for which the kinetic equations read as 
\begin{equation*}\label{gas}
\begin{array}{c}
f_t+(sf)_x=0,\\
\ \\
s(\eta)=\eta-\mathop{\int}_{0}^{\infty}af(\mu)[s(\mu)-s(\eta)]\ d\mu,
\end{array}
\end{equation*}
where $f(\eta)=f(\eta, x, t)$ is the distribution function, $s(\eta)=s(\eta, x, t)$ is the associated transport velocity and $a $ is a constant.

Under the delta-functional ansatz $$f(\eta, x, t)=\sum_{i=1}^{m}u^i(x, t)\ \delta(\eta-\eta^i(x, t)),$$and a transformation of the field variables, the integro-differential system is reduced to a quasilinear one of Toeplitz type (see \cite{EKPZ,FP} for further details):
\begin{equation}\label{hrred}
\begin{cases}u^1_{\alpha,t}=\lambda^\alpha\, u^1_{\alpha,x}+\mu^\alpha \,u^2_{\alpha,x}\\
u^2_{\alpha,t}=\lambda^\alpha\, u^2_{\alpha,x}
\end{cases}\qquad  \alpha =1,2, \dots m.
\end{equation}

\bigskip

Let us now fix $m=2$, i.e. we deal with a $4\times 4$ quasilinear system. For sake of simplicity (due to the large computations needed in this Section) we choose to adopt the notation identifying the field variables  with $u^1,u^2,u^3$ and $u^4$. Then, the matrix $A$ of velocities has the following entries

 \begin{equation*}
      \lambda^{1}=- \dfrac{u^{3} u^{2}+ a u^{4}}{u^{3}+a}, \quad \lambda^{2}=-\dfrac{u^{1} u^{4}+ a u^{2}}{u^{1}+a}, \qquad \mu^{1}=\dfrac{u^{1} u^{3}-a^{2}}{u^{3}+a}, \quad \mu^{2}=\dfrac{u^{1} u^{3}-a^{2}}{u^{1}+a},
  \end{equation*}

and the system is endowed with the constraints
  \begin{equation}
      u_{x}^{2}=\varphi^{1}(u^{1}, u^{2}, u^{3}, u^{4}) \qquad u_{x}^{4}=\varphi^{2}(u^{1}, u^{2}, u^{3}, u^{4}) .
      \label{vincn3}
  \end{equation}
 
Studying the compatibility  as in the proof of Theorem \ref{thm4}, we get that the constraints need to satisfy
\begin{equation}
   \varphi^1= \dfrac{f^1(u^1_1,u^2_1,u^2_2)}{\lambda^2 - \lambda^1},\qquad 
   \varphi^2= \dfrac{f^2(u^2_1, u^1_2,u^2_2)}{\lambda^2 - \lambda^1},
   \label{compahardord}
   \end{equation}
where $f^1$ and $f^2$ are functions to be specified later. Moreover, we find the additional conditions
 \begin{subequations}\label{compatibilityhardord12}
 \begin{align}&  \lambda^1_{u^2_1} (\varphi^1)^2 + \varphi^2 (\lambda^1 \varphi^1)_{u^2_2}= \varphi^1_{u^1_1} \varphi^1 \mu^1
+ \mu^2 \varphi^1_{u^1_2} \varphi^2 + \lambda^2 \varphi^1_{u^2_2} \varphi^2,  \label{compatibilityhardord1}\\
& \lambda^2_{u^2_2} (\varphi^2)^2+\varphi^1 (\lambda^2 \varphi^2)_{u^2_1}=  \varphi^2_{u^1_1} \varphi^1 \mu^1+ \mu^2 \varphi^2_{u^1_2} \varphi^2+\lambda^1 \varphi^2_{u^2_1} \varphi^1  .\label{compatibilityhardord2}
\end{align}
\end{subequations} 
Then, taking the constraints \eqref{vincn3} into account, the equations \eqref{hrred} assume the form.
\begin{equation}
\label{sistemaiper}
    \begin{cases}
        u_{t}^{1}- \lambda^{1} u_{x}^{1}=\mu^{1} \varphi^{1}\\
        u_{t}^{2}-\lambda^{1} u_{x}^{2}=0\\
        u_{t}^{3}-\lambda^{2} u_{x}^{3}= \mu^{2} \varphi^{2}\\
        u_{t}^{4}-\lambda^{2} u_{x}^{4}=0.
    \end{cases}
\end{equation}
The resulting system is hyperbolic, with eigenvalues
$\nu^{(1,2)}= -\lambda^{1}$ and $ \nu^{(3,4)}=-\lambda^{2}$
and left eigenvectors
\begin{equation*}
    l^{(1)}=(1,0,0,0), \quad l^{(2)}=(0,1,0,0), \quad l^{(3)}=(0,0,1,0) \quad l^{(4)}=(0,0,0,1).
\end{equation*}
Once that the system is rewritten in a hyperbolic form, it is a well-known fact (see \cite{Zhi}) that the first order constraints are quasilinear and take form
\begin{equation*}
    l^{(i)} \cdot \textbf{u}_x= \varphi^{i}, \qquad i=1, \dots,  \, 4.
\end{equation*}

This means that if we can choose between
$u_{x}^{3}= \varphi^{3}$ or 
    $u_{x}^{1}=\varphi^{4}$.
Withouth loss of generality, we choose the first option so that we deal with the three constraints
\begin{equation*}
    u_{x}^{2}=\varphi^{1}, \qquad u_{x}^{3}=\varphi^{3}, \qquad u_{x}^{4}=\varphi^{2}.
\end{equation*}
We remark that when the number of the constraints is $N-1$ (given $N$ the number of the governing equations), the integration of the original systems of PDEs is reduced to that of a ODEs system written along the characteristic curves associated to the hyperbolic system under interest \cite{2a}.
With this approach we find supplementary compatibility conditions given by
\begin{equation}
    \varphi^{3}=\dfrac{f^{3}(u^{2}, u^{3}, u^{4})}{\lambda^{2}-\lambda^{1}} \label{compatibilityhardord3}
\end{equation}
along with
\begin{align}\begin{split}
   &\mu^{1} \varphi^{1} \varphi_{u^{1}}^{3}+ \mu^{2} \varphi^{2} \varphi_{u^{3}}^{3}-  \lambda_{u^{2}}^{2} \varphi^{1} \varphi^{3}-\lambda_{u^{4}}^{2} \varphi^{2} \varphi^{3}=\\&\hphantom{ciaocioaio}(\lambda^{2}-\lambda^{1}) \varphi^{1} \varphi_{u^{2}}^{3}+ (\mu^{1} \varphi^{2})_{u^{2}} \varphi^{1}+ (\mu^{1} \varphi^{2})_{u^{3}} \varphi^{3}+ + (\mu^{1} \varphi^{2})_{u^{4}} \varphi^{2}.
   \label{compatibilityhardrod4}\end{split}
\end{align}
\bigskip

Substituting the form of $\varphi^{1}, \varphi^{2}, \varphi^{3}$ given in \eqref{compahardord}, \eqref{compatibilityhardord3},  into \eqref{compatibilityhardord1}, \eqref{compatibilityhardord2}, \eqref{compatibilityhardrod4}, we get that $f^1$, $f^2$ and $f^3$ satisfy
 \begin{subequations}
 \label{33}
\begin{align}
       &(u^{1}+a) f_{u^{1}}^{1}f^{1}+ (f^{1})^2+ (u^{2}- u^{4}) f_{u^{4}}^{1} f^{2} +f^{1} f^{2}=0 ,\label{compatibilitya}\\
   & (u^{3}+a) f_{u^{3}}^{2} f^{2}+ (f^{2})^{2}-  (u^{2}- u^{4})f_{u^{2}}^{2}f^{1} -f^{2} f^{1}=0, \label{compatibilityb}\\
    \begin{split}&(u^{3}+a)( f^{2}f_{u^{3}}^{3}- f^{3} f_{u^{3}}^{2})-  (u^{2}- u^{4}) f_{u^{2}}^{3} f^{1}+f^{3}f^{1}+ \\
    &\hphantom{cicociaocioa}- \frac{u^{3}+a}{u^{2}-u^{4}} [(u^{2}- u^{4})f_{u^{2}}^{2}f^{1}  -f^{2} f^{1}+  (u^{2}- u^{4}) f_{u^{4}}^{2} f^{2}  +(f^{2})^{2}]=0.
    \label{compatibilityc}\end{split}
\end{align}
 \end{subequations}
We can rewrite the system \eqref{sistemaiper} along the characteristics as

\begin{equation}
    \begin{cases}
        u_{t}^{1}-\lambda^{1} u_{x}^{1}= \mu^{1} \varphi^{1}\\
        u_{t}^{2}-\lambda^{1} u_{x}^{2}=0\\
        u_{t}^{3}-\lambda^{1} u_{x}^{3}= (\lambda^{2}-\lambda^{1}) \varphi^{3}+ \mu^{2} \varphi^{2}\\
        u_{t}^{4}- \lambda^{1} u_{x}^{4}=(\lambda^{2}-\lambda^{1}) \varphi^{2}
    \end{cases},
    \label{sistemacaratteristiche}
\end{equation}
with the additional relations \eqref{33}  on the structure of the constraints. We distinguish  two cases ($f_{u^{3}}^{2} \neq 0$ and $f_{u^{3}}^{2}=0)$, so that we are able to find two solutions to the hard rod case.

\bigskip

\textbf{Case 1.}
As a first case, we assume $f_{u^{3}}^{2} \neq 0$. Hence, from \eqref{compatibilitya} we can deduce
\begin{align*}
 & (u^{2}-u^{4})  f_{u^{4}}^{1}+ f^{1}=0\\
  &(u^{1}+a) f_{u^{1}}^{1}+f^{1} =0
\end{align*}
from which by integration we get
\begin{equation}
    f^{1}= \dfrac{c_{1}(u^{2}) (u^{2}- u^{4})}{u^{1}+a},
\end{equation}
where $c_1$ is an arbitrary function of its argument.
At this point, since $f_{u^{1}}^{1} \neq 0$, we get from \eqref{compatibilityb}
 \begin{align*}
      &(u^{2}- u^{4}) f_{u^{2}}^{2}- f^{2}=0\\
      &(u^{3}+a) f^{2}_{u^{3}}+f^{2}=0
 \end{align*}
 In this way, we can deduce the form of $f^{2}$, given by
 \begin{equation*}
     f^{2}= -\dfrac{ (u^{2}-u^{4}) c_{2}(u^{4})}{u^{3}+a},
 \end{equation*}
 where $c_2$ is an arbitrary function of its argument that we will choose constant, so that $c_2(u^4)=k \in \mathbb{R}$.\\
 Taking the form of $f^{1}$ and $f^{2}$ into account, we can rewrite the last compatibility condition as
 \begin{equation}
    (u^{3}+a) f_{u^{3}}^{3} f^{2}+f^{3} f^{2}=(u^{2}-u^{4}) f_{u^{2}}^{3} f^{1}- f^{3} f^{1}
    \label{ultima}
 \end{equation}
 To produce an exact solution, we simply assume $f^{3}=0.$ With this requirement, \eqref{ultima} is trivially satisfied. Under this assumption, system \eqref{sistemacaratteristiche} can be easily integrated and we obtain
\begin{equation}
\begin{cases}\label{sol2u}
    
    u^{1}=c_1\, u_{0}^{2}\; t+ u_{0}^{1}\\
    u^{2}= u_{0}^{2}\\
    u^{3}= -k t+h\\
    u^{4}=\dfrac{u_{0}^{2} \, k\,  t -(h+a) u_{0}^{4}}{k t- h-a}
    \end{cases},
\end{equation}
where $h$ is an arbitrary constant, $u^1_0=u^1_0(\sigma)$, $u^2_0=u^2_0(\sigma)$, $u^4_0=u^4_0(\sigma)$ are arbitrary functions of the variable $\sigma$ that is given from the characteristic equation
\begin{equation*}
    \dfrac{d x}{ d t}=-\lambda^{1} \label{eqcar}
\end{equation*}
and hence it is implicitly defined by
\begin{equation*}
x= \dfrac{k u_{0}^{2}(\sigma) t+ h u_{0}^{2}(\sigma)+a u_{0}^{4}(\sigma)}{-kt+h+a} t+ \sigma.
\end{equation*}
 Substituting the solutions into the differential constraints, we obtain two relations among the functions $u_{0}^{1}(\sigma),u_{0}^{2}(\sigma), u_{0}^{4}(\sigma)$ and $c_1(u_{0}^{2})$, i.e.

 \begin{equation*}
 \begin{cases}
        u_{0}^{2' }= c_1(u_{0}^{2}) \dfrac{a+ k_1}{h u_{0}^{1}-a^{2}} \vspace{0.25cm}\\
          u_{0}^{ 4'}=- \dfrac{k(u_{0}^{1}+a)}{h u_{0}^{1}-a^{2}}
         \end{cases},
 \end{equation*}
 while the constraint for $u_{0}^{3}$ is identically satisfied.\\
 This means that we are able to find a solution to the hard rod case depending on two arbitrary
functions.

\bigskip 

 \textbf{Case 2.}
 We now consider the case $f_{u^{3}}^{2}=0.$
Here,  it is easy to notice $f^{1}= f^{1}(u^{2}, u^{4})$.
Hence, the system of compatibility conditions becomes
\begin{align}
  \begin{split}
        &(f^{2})^2-(u^{2}-u^{4}) f_{u^{2}}^{2}f^{1}+ f^{2}f^{1}=0\\
        &(f^{1})^{2}+(u^{2}-u^{4}) f_{u^{4}}^{1}f^{2}- f^{1}f^{2}=0\\
      & (u^{3}+a) f^{2} f_{u^{3}}^{3}  -(u^{2}-u^{4}) f_{u^{2}}^{3}f^{1}+f^{3}f^{1}- 2  \dfrac{u^{3}+a}{u^{2}-u^{4}}(f^{2})^{2}- (u^{3}+a) f^{2} f_{u^{4}}^{2}=0
        \end{split}\label{condizionidicompatibilità0}
\end{align}
At this point, we require $f_{u^{3}}^{3}=0$, so that we are able to integrate the compatibility conditions above, from which
\begin{align*}
        f^{1}&=\dfrac{c_1(u^{2})^2 (u^{2}-u^{4})^2 }{c_1(u^{2}) +(u^{2}-u^{4}) c_{1}^{'}(u^{2})} \\
        f^{2}&=(u^{2}-u^{4})^{2} c_{1}(u^{2})\\
        f^{3}&= (u^{2}-u^{4}) c_{2}(u^{4})
    ,
\end{align*}
where $c_i, \, i=1,2$ are arbitrary functions of their arguments.
In order to produce a solution, we will make the further assumption $c_2(u^{4})=k \in \mathbb{R}.$
Now, we are able to integrate system \eqref{sistemacaratteristiche} and we obtain

\medskip

\begin{equation}
    \begin{cases}\label{sol2v}
        u^{1}= \dfrac{c_{1}^{2}(u_{0}^{2}) (u_{0}^{1}+a) (u_{0}^{2}-u_{0}^{4})t }{c_{1}(u_{0}^{2}) + (u_{0}^{2}- u_{0}^{4}) c_{1}'(u_{0}^{2})}+u_{0}^{1}\vspace{0.25cm}\\
        u^{2}=u_{0}^{2} \vspace{0.25cm}\\
        u^{3}= (u_{0}^{2}-u_{0}^{4})[c_1(u_{0}^{2}) (u_{0}^{3}+a)+k]t +u_{0}^{3} \vspace{0.25cm}\\
        u^{4}=\dfrac{u_{0}^{4}+ u_{0}^{2} (u_{0}^{2}-u_{0}^{4}) c_1(u_{0}^{2})t}{c_1(u_{0}^{2}) (u_{0}^{2}- u_{0}^{4})t+1}
    \end{cases},
\end{equation}
\medskip

where $u^1_0=u^1_0(\sigma)$,  $u^2_0=u^2_0(\sigma)$, $u^3_0=u^3_0(\sigma)$, $u^4_0=u^4_0(\sigma)$ are arbitrary functions and $\sigma$ is obtained from the characteristic equation \eqref{eqcar} and so it is implicily defined by
\smallskip
\begin{equation*}
    x=\dfrac{a}{k} \ln \left( \dfrac{(u_{0}^{3}+a)(1+ c_1(u_{0}^{2}) u_{0}^{2} t-  c_1(u_{0}^{2}) u_{0}^{4} t)}{u_{0}^{3}+a+( c_1(u_{0}^{2}) u_{0}^{3}+ c_1(u_{0}^{2})a+k)(u_{0}^{2}- u_{0}^{4}) t}\right)+ u_{0}^{2}t + \sigma.
\end{equation*}
\medskip

In the end, we substitute the solutions into the differential constraints. We find that the functions
$u_{0}^{1}(\sigma),u_{0}^{2}(\sigma), u_{0}^{3}(\sigma), u_{0}^{4}(\sigma)$ and $c_{1}(u_{0}^{2})(\sigma)$ must satisfy
\medskip

\begin{equation*}
    \begin{cases}
        u_{0}^{2'}=\dfrac{c_{1}^{2}(u_{0}^{2})(u_{0}^{1}+a)(u_{0}^{3}+a)(u_{0}^{2}- u_{0}^{4})}{(u_{0}^{1}u_{0}^{3}-a^{2})(c_1(u_{0}^{2}+c_{1}^{'}(u_{0}^{2}) (u_{0}^{2}-u_{0}^{4}))} \vspace{0.25cm}\\
        u_{0}^{3'}=\dfrac{k(u_{0}^{1}+a)(u_{0}^{3}+a)}{u_{0}^{1}u_{0}^{3}-a^{2}}\vspace{0.25cm}\\
        u_{0}^{4'}=\dfrac{c_{1}(u_{0}^{2})(u_{0}^{1}+a)(u_{0}^{3}+a)(u_{0}^{2}- u_{0}^{4})}{u_{0}^{1}u_{0}^{3}-a^{2}}
    \end{cases}
\end{equation*}
We can conclude that, also in this second case, we obtain a solution up to two arbitrary functions.
 \medskip

\noindent \textbf{Remark.} To simplify our computations, some assumptions have been taken in both the cases. In particular, $f_3=0$ in the first case and $f^3_{u^3}=c_{2,u^4}=0$ in the second. These additional constraints, together with the third quasilinear one $u^3_x=\varphi^3$, seem to reduce the number of arbitrary functions to two (from the ideal objective of four). 

In spite of this, we present two exact solutions to a relevant physical problem that, as far as the authors know, has very few analogues in the literature (one of them was obtained in \cite{FP}, see formula (13)). In particular, a concrete choice of the arbitrary functions could represent a novel result in soliton gas theory. This is the main topic of a forthcoming paper. 

\section{Hamiltonian structures}\label{5}
The Hamiltonian property for Jordan-block systems in Toeplitz form was deeply investigated by E.V. Ferapontov et al. in \cite{VerFer1}. In this Section, we show some analogies bewteen the Hamiltonian structure of such systems and quasilinear differential constraints. 

First of all, let us recall that a quasilinear system $u^i_t=A^i_j(\textbf{u})u^j_x$ is said Hamiltonian if there exists a Hamiltonian operator $P^{ij}=a^{ij\sigma}\partial_\sigma$ and a Hamiltonian functional $H=\int{h(u)\, dx}$ such that 
\begin{equation*}
    u^i_t=A^i_j(\textbf{u})u^j_x=P^{ij}\frac{\partial h}{\partial u^j} \qquad i=1,2,\dots n.
\end{equation*}
In our case, we ask the operator to be of Dubrovin-Novikov form $P^{ij}=g^{ij}\partial_x+b^{ij}_ku^k_x$, where $g^{ij},b^{ij}_k$ depend on the field variables only. The Hamiltonianity conditions for such an operator require, in the non-degenerate case $\det g\neq 0$,  $g$ to be a flat cometric and $b^{ij}_k=-g^{is}\Gamma^j_{sk}$, where $\Gamma^i_{sj}$ are the Christoffel symbols of $g$ \cite{DN83}. In addition, in \cite{Tsarev, Tsarev1}, S.I. Tsarev proved that the following conditions are necessary and sufficient for a quasilinear systems to be Hamiltonian 
\begin{equation}\label{tsa}
    g_{is}A^s_j=g_{js}A^s_i,\qquad \nabla_iA^j_k=\nabla_kA^j_i, \qquad i,j,k=1,2,\dots n,
\end{equation}where $\nabla$ is the covariant derivative with respect to the Levi-Civita connection of the flat metric $g$. 

Conditions \eqref{tsa} were used in \cite{VerFer1} to prove that  

\bigskip 

\begin{theorem}
Suppose that the matrix $A$ of system \eqref{parsys} has block-diagonal form with
several blocks of Toeplitz type with distinct eigenvalues.
Then the existence of Dubrovin-Novikov Hamiltonian structure implies linear degeneracy.
\end{theorem}

\bigskip

Then, specifying Tsarev's conditions for $2\times 2$
Jordan-block systems we obtain the following necessary and sufficient condition on the Hamiltonianity:

\bigskip

\begin{theorem}
    For $2\times 2$ hydrodynamic-type systems in Jordan-block form, the following conditions are equivalent
    \begin{itemize}
        \item[(\text{i})] The system is linearly degenerate;
         \item[(ii)] The system admits quasilinear constraints $u^2_x=\varphi^1(t,x,\textbf{u})$;
        \item[(iii)] The systems admits Hamiltonian formalism in Dubrovin-Novikov sense.
       
    \end{itemize}
\end{theorem}
\begin{proof}
The equivalence of conditions (i) and (ii) is proved in Theorem \ref{thm3}. 

Let us now investigate the equivalence of (i) and (iii). At first, from the first set of Tsarev's conditions \eqref{tsa} we get that $g$ is a metric in Hankel form:
\begin{equation*}
    g_{ij}=\begin{pmatrix}
        0&g_{12}(u^1,u^2)\\g_{12}(u^1,u^2)&g_{22}(u^1,u^2)
    \end{pmatrix}
\end{equation*}
Whereas, the second set reads as follows
\begin{equation*}\label{qst}
\nabla_1A^2_1-\nabla_2A^2_2=\frac{\partial \lambda}{\partial u^1}\qquad \nabla_1A^1_2-\nabla_2A^1_1=\frac{\partial \mu}{\partial u^1}+\Gamma^{1}_{11}\, \mu-\frac{\partial \lambda}{\partial u^2}.
\end{equation*}Note that these are the only non trivial conditions. Then, by applying the formula for Christoffel symbols:
{\begin{align*}
g_{12}(u^1,u^2)&=f^0(u^2)e^{-\displaystyle\int{\frac{\frac{\partial \mu}{\partial u^1}-\frac{\partial \lambda}{\partial u^2}}{\mu}du^1}}\\&= f^0(u^2)\left(e^{-\displaystyle \int{\frac{\frac{\partial \mu}{\partial u^1}}{\mu}\, du^1}}e^{\displaystyle\int{\frac{\frac{\partial \lambda}{\partial u^2}}{\mu}\, du^1}}\right)\\
&=\frac{f^1(u^2)}{\mu(u^1,u^2)}e^{\displaystyle\int{\frac{\frac{\partial \lambda}{\partial u^2}}{\mu}\, du^1}},
\end{align*}}
where $f^0(u^2)$, $\eta(u^2)$ are arbitrary functions, and $f^1(u^2)=f^0(u^2)e^{\eta(u^2)}$.

Finally, defining $g_{22}(u^1,u^2)=0$ then we obtain exactly a flat metric and Tsarev conditions are satisfied if and only if $\lambda$ does not depend on $u^1$, i.e. if and only if the system is linearly degenerate (Theorem \ref{thm0}).
\end{proof}

\bigskip

The previous Theorem leads to a deeper consideration on the relation between the existence of the Hamiltonian structure of parabolic systems and the existence of quasilinear compatible differential constraints. Indeed, one can directly obtain that for linear degenerate systems of Jordan-block type in $2$ components the function
\begin{equation*}
    \theta(u^1,u^2)=f(u^1)e^{\displaystyle \int{\frac{\frac{\partial \lambda}{\partial u^2}}{\mu}\, du^1}},
\end{equation*}
uniquely defines both the Hamiltonian structure with a non-degenerate flat metric $$g=\frac{\partial \theta}{\partial u^1}\, (du^1\otimes du^2+du^2\otimes du^1)$$ and the compatible quasilinear constraint by $u^2_x=\theta(u^1,u^2)$. We plan to focus our attention to the investigation of such an interesting property in another future work. 

\section{Conclusions}

In this paper, we applied the method of differential constraints to parabolic quasilinear systems of first order PDEs with a upper triangular Toeplitz form. In particular, we investigated the connections between quasilinear differential constraints and linearly degenerate systems and we focused on solutions of some systems coming from physical theories. We list here briefly our main results:
\begin{itemize}
    \item We proved that for systems of $2$ components the only differential constraints compatible with the systems are quasilinear. In addition, we compared the Hamiltonian property for $2\times 2$ systems to the existence of quasilinear constraints, proving the equivalence of the properties;
    \item We showed that linear degeneracy is a necessary condition for systems in higher number of components to admit quasilinear differential constraints;
    \item We mapped every system in $2$ components into one we are able to integrate and find solutions depending on $2$ arbitrary functions, i.e. finding the most general integral;
    \item We solved two examples arising from the theory of WDVV equations and principal hierarchies from Frobenius manifolds. Moreover, we found a solution of the El's kinetic equation in the hard rod case under delta-functional ansatz, by a reduction of the system into a hyperbolic one.
\end{itemize}

Several open questions arise from the present results. As an example, we wonder if a direct connection between the existence of quasilinear differential constraints and the Hamiltonian property exists. In addition, it would be of interest the study of differential constraints which are genuinely nonlinear and their geometric interpretation in the contest of hydrodynamic systems. Finally, the authors will devote a future work on the integration of other examples of kinetic equations under delta-functional reduction, wondering if such a procedure gives explicit solutions also in other cases of soliton gas equations (as the KdV, the sinh-Gordon or the DNLS ones).

\bigskip

{ \subsection*{Acknowledgements} The authors are extremely thankful to N. Manganaro for the possibility to work on this problem, for his comments and suggestions and for his concrete help. We also thank E.V. Ferapontov and T. Congy for stimulating discussions and interesting remarks. We finally acknowledge the financial support of GNFM of the Istituto Nazionale di Alta Matematica.  PV is partially funded by the research project Mathematical Methods in Non- Linear Physics (MMNLP) by the Commissione Scientifica Nazionale – Gruppo 4 – Fisica Teorica of the Istituto Nazionale di Fisica Nucleare (INFN) and by "Borse per viaggi all'estero" of the Istituto Nazionale di Alta Matematica, which permitted to visit the Geometry and Physics group of Loughborough University.}

\end{document}